\documentclass[runningheads]{llncs}
\usepackage[T1]{fontenc}
\usepackage{graphicx}
\usepackage{xspace}
\usepackage{amsmath}
\usepackage{amsfonts}
\spnewtheorem{assumption}{Assumption}{\bfseries}{\normalfont}
\usepackage[labelformat=simple]{subcaption}
 
\usepackage{svg}
\usepackage{interval}
\usepackage{algorithm,algpseudocode}
\usepackage{pgfplots,tikz}
\pgfplotsset{compat=1.18}
\usepackage{multirow}
\usepackage{booktabs}
\usepackage{hyperref}
\usepackage[most]{tcolorbox}
\usepackage{todonotes}

\title{Efficient Shield Synthesis \\ via State-Space Transformation}

\author{Asger Horn Brorholt \and
Andreas Holck H{\o}eg-Petersen \and
\\
Kim Guldstrand Larsen \and
Christian~Schilling}

\authorrunning{Brorholt, H{\o}eg-Petersen, Larsen, and Schilling}

\institute{Aalborg University, 9220 Aalborg, Denmark
\email{\{asgerhb,ahhp,kgl,christianms\}@cs.aau.dk}}

\newcommand{\figref}[1]{Fig.~\ref{fig:#1}\xspace}

\newcommand{\uppaalstratego}{\textsc{Uppaal Stratego}\xspace}
\newcommand{\uppaal}{\textsc{Uppaal}\xspace}

\newenvironment{psmallmatrix}{\left(\begin{smallmatrix}}{\end{smallmatrix}\right)}
\newcommand{\twovec}[2]{\begin{psmallmatrix} #1 \\ #2 \end{psmallmatrix}}

\newcommand{\RR}{\mathbb{R}}

\newcommand{\powerset}[1]{\ensuremath{2^{#1}}\xspace}
\DeclareMathOperator{\atan}{atan2}
\newcommand{\I}{\ensuremath{\mathcal{I}}\xspace}
\newcommand{\grid}{\ensuremath{\mathcal{G}}\xspace}
\newcommand{\cell}{\ensuremath{C}\xspace}
\newcommand{\act}{\ensuremath{\mathit{Act}}\xspace}
\newcommand{\suc}{\ensuremath{\delta}\xspace}
\newcommand{\strategy}{\ensuremath{\sigma}\xspace}
\newcommand{\safe}{\ensuremath{\varphi}\xspace}
\newcommand{\controllablecells}[1][\safe]{\ensuremath{\mathcal{C}_{#1}}\xspace}
\newcommand{\controllablecellstransform}[1][\safe]{\ensuremath{\controllablecells[#1]^{\f}}\xspace}
\newcommand{\traj}{\ensuremath{\xi}\xspace}
\newcommand{\f}{\ensuremath{f}\xspace}
\newcommand{\finv}{\ensuremath{\f^{-1}}\xspace}
\newcommand{\getcellsinner}[1]{\ensuremath{\lfloor #1 \rfloor}_{\grid}\xspace}
\newcommand{\getcellsouter}[1]{\ensuremath{\lceil #1 \rceil}_{\grid}\xspace}

\algnewcommand{\To}{\textbf{To }}
\algnewcommand{\Input}{\item[\textbf{Input:}]}%
\algnewcommand{\Output}{\item[\textbf{Output:}]}%

\newtcolorbox{statespacebox}{
    title={State space},
    colback=clouds,
    colframe=wetasphalt,
    sharp corners,
    left=4pt,
    right=4pt,
    top=4pt,
    bottom=4pt,
    leftrule=1pt,
    rightrule=1pt,
    toprule=1pt,
    bottomrule=1pt
}

\definecolor{turquoise}{HTML}{1ABC9C}
\definecolor{emerald}{HTML}{2ECC71}
\definecolor{peterriver}{HTML}{3498DB}
\definecolor{amethyst}{HTML}{9B59B6}
\definecolor{wetasphalt}{HTML}{34495E}

\definecolor{greensea}{HTML}{16A085}
\definecolor{nephritis}{HTML}{27AE60}
\definecolor{belizehole}{HTML}{2980B9}
\definecolor{wisteria}{HTML}{8E44AD}
\definecolor{midnightblue}{HTML}{2C3E50}

\definecolor{sunflower}{HTML}{F1C40F}
\definecolor{carrot}{HTML}{E67E22}
\definecolor{alizarin}{HTML}{E74C3C}
\definecolor{clouds}{HTML}{ECF0F1}
\definecolor{concrete}{HTML}{95A5A6}

\definecolor{orange}{HTML}{F39C12}
\definecolor{pumpkin}{HTML}{D35400}
\definecolor{pomegranate}{HTML}{C0392B}
\definecolor{silver}{HTML}{BDC3C7}
\definecolor{asbestos}{HTML}{7F8C8D}

\colorlet{successor}{blue!50!red}

\begin{document}

\maketitle

\begin{abstract}
We consider the problem of synthesizing safety strategies for control systems, also known as shields. Since the state space is infinite, shields are typically computed over a finite-state abstraction, with the most common abstraction being a rectangular grid. However, for many systems, such a grid does not align well with the safety property or the system dynamics. That is why a coarse grid is rarely sufficient, but a fine grid is typically computationally infeasible to obtain. In this paper, we show that appropriate state-space transformations can still allow to use a coarse grid at almost no computational overhead. We demonstrate in three case studies that our transformation-based synthesis outperforms a standard synthesis by several orders of magnitude. In the first two case studies, we use domain knowledge to select a suitable transformation. In the third case study, we instead report on results in engineering a transformation without domain knowledge.

\end{abstract}

\section{Introduction}

Cyber-physical systems are ubiquitous in the modern world.
A key component in these systems is the digital controller.
Many of these systems are safety critical, which motivates the use of methods for the automatic construction of controllers.
Unfortunately, this problem is intricate for any but the simplest systems~\cite{LewisVS12,DoyleFT13}.

Two main methods have emerged.
The first method is \emph{reinforcement learning} (RL)~\cite{BusoniuBTKP18}, which provides convergence to an optimal solution.
However, the solution lacks formal guarantees about safety.
The second method is \emph{reactive synthesis}, which constructs a nondeterministic control strategy that is guaranteed to be safe.
However, the solution lacks optimality for secondary objectives.

\begin{figure}[t]
    \centering
    \input{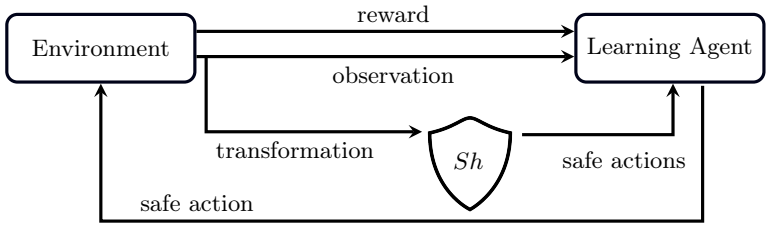}
    \caption{Reinforcement learning under a shield with a transformation~$\f$.}
    \label{fig:shielding}
\end{figure}

Due to their complementary strengths and drawbacks, these two methods have been successfully combined in the framework of \emph{shielding}~\cite{DavidJLLLST14,BloemKKW15} (cf.\ \figref{shielding}).
Through reactive synthesis, one first computes a nondeterministic control strategy called a \emph{shield}, which is then integrated in the learning process to prevent unsafe actions.
This way, safety is guaranteed and, at the same time, RL can still provide optimality with respect to the secondary objectives.

In this work, we focus on the first step: synthesis of a shield.
For infinite state spaces, we employ an abstraction technique based on state-space partitioning, where we consider the common case of a \emph{hyperrectangular grid}~\cite{Girard12,Tabuada09}.
This grid induces a finite-state two-player game, from which we can then construct the most permissive shield with standard algorithms.
The downside of a grid-based approach is that a grid often does not align well with the dynamics of the system, which causes the game to have many transitions and thus results in an overly conservative control strategy.
To counter this effect, one can refine the partition, but this has severe computational cost and quickly becomes infeasible.

\medskip

The key insight we follow in this work is that a \emph{state-space transformation} can yield a new state space where the grid aligns much better with the system dynamics.
As we show in three case studies, the transformation allows to reduce the grid significantly, often by several orders of magnitude.
In extreme cases, no grid may exist for synthesizing a control strategy in the original state space, while a simple grid suffices in the transformed state space.

We show that our transformation-based shield synthesis is sound, i.e., the guarantees of the shield transfer to the original system.
Moreover, our experiments demonstrate that the transformation does not reduce the performance of the final controller (in fact, the performance slightly increases).

Our implementation is based on our previous work on sampling-based shield synthesis~\cite{BrorholtJLLS23}.
The present work integrates nicely with such a sampling-based method, but also generalizes to set-based methods.

For the first two case studies, we employ domain knowledge to derive a suitable transformation.
For the third case study, we instead apply a new heuristic method to synthesize a suitable transformation.

\subsection{Related Work}

Abstraction-based controller synthesis is a popular approach that automatically constructs a controller from a system and a specification~\cite{Girard12,Tabuada09}.
The continuous dynamics are discretized and abstracted by a symbolic transition system, for which then a controller is found.
The most common abstraction is a regular hyperrectangular (or even hypercubic) grid.
The success of this approach depends on the choice of the grid cells' size.
If too small, the state space becomes intractably large, while if too large, the abstraction becomes imprecise and a controller may not be found.
While the cell size can be optimized~\cite{WeberRR17}, a fixed-sized grid is often bound to fail.
Instead, several works employ multiple layers of different-sized grids in the hope that coarser cells can be used most of the time but higher precision can still be used when necessary~\cite{GirardGM16,HsuMMS18a}.
In this paper, we follow an orthogonal approach.
We argue that a hyperrectangular grid is often inappropriate to capture the system dynamics and specification.
Nevertheless, we demonstrate that, often, a (coarse) grid is still sufficient when applied in a different, more suitable state space.
Orthogonal to our approach are recent efforts to employ ellipsoidal grid-like abstractions, which can also be computed lazily~\cite{EgidioLJ22,CalbertELJ24,CalbertBLJ24}.

In this work, we do not synthesize a full controller but only a nondeterministic safety strategy, which is known as a shield.
We then employ this shield as a guardrail in reinforcement learning (cf.\ \figref{shielding}) to limit the choices available to the agent so that the specification is guaranteed.
This is a known concept, which is for instance applied in the tool \uppaalstratego~\cite{stratego} and was popularized by Bloem et al.~\cite{BloemKKW15,AlshiekhBEKNT18,JansenKJSB20}.
A similar concept is safe model predictive control~\cite{BastaniL21,WabersichZ21}.

Our motivation for applying a state-space transformation is to better align with a grid, and ultimately to make the synthesis more scalable.
In that sense, our work shares the goal with some other influential concepts.
In abstract interpretation, the transformation is the abstraction function and its inverse is the concretization function, which together form a Galois connection~\cite{CousotC77}.
In our approach, the grid introduces an abstraction, but our additional transformation preserves information unless it is not injective.
Another related concept is model order reduction, where a system is transformed to another system of lower dimensionality to simplify the analysis~\cite{SchildersHR08}.
This reduction is typically approximate, which loses any formal guarantees.
However, approaches based on (probabilistic) bisimulation~\cite{LarsenS91} still allow to preserve a subspace and transfer the results to the original system.
These approaches, also called lumpability, use linear transformations and have been successfully applied to Markov chains~\cite{Buchholz94}, differential equations~\cite{BacciBLTTV21}, and quantum circuits~\cite{JimenezPastorLTT24}.
In contrast, while we do not put any restrictions on our transformations, we advocate for injective transformations in our context; this is because we also need to compute the preimage under the transformation, which otherwise incur additional approximation error.

\paragraph{Outline.}
The remainder of the paper is structured as follows.
In Section~\ref{sec:preliminaries}, we recall central concepts underlying partition-based shielding.
In Section~\ref{sec:shielding}, we discuss how state-space transformations can be used for shield synthesis over grid-based partitions.
In Section~\ref{sec:experiments}, we show experimental results in three case studies.
Finally, we conclude the paper in Section~\ref{sec:conclusion}.

\section{Preliminaries}\label{sec:preliminaries}


\paragraph{Intervals.}

Given bounds~$\ell, u \in \RR$ with $\ell \le u$, we write~$\I = \rinterval{\ell}{u} \subseteq \RR$ for the corresponding (half-open) \emph{interval} with \emph{diameter}~$u - \ell$.

\paragraph{Set extension.}

Given a function~$f \colon S \to T$, the \emph{set extension} is~$f \colon \powerset{S} \to \powerset{T}$ with~$f(X) = \bigcup_{s \in X} \{f(s)\}$ for any subset~$X \subseteq S$.
The extension generalizes to functions with further arguments, e.g., $g(X, y) = \bigcup_{s \in X} \{g(s, y)\}$.
Set extension is monotonic, i.e., $f(X) \subseteq f(X')$ whenever~$X \subseteq X'$.

\paragraph{Control systems.}

In this work, we consider discrete-time control systems.
Formally, a control system~$(S, \act, \suc)$ is characterized by a bounded $d$-dimensional state space~$S \subseteq \RR^d$, a finite set of (control) actions~$\act$, and a \emph{successor function}~$\suc \colon S \times \act \to \powerset{S}$, which maps a \emph{state}~$s \in S$ and a \emph{control action}~$a \in \act$ to a set of successor states (i.e., $\suc$ may be nondeterministic).
Often, $\suc$ is the solution of an underlying continuous-time system, measured after a fixed control period, as exemplified next.

\begin{example}\label{ex:oscillator1}
    Consider a bivariate harmonic oscillator over the state space~$S = \rinterval{-2}{2} \times \rinterval{-2}{2} \subseteq \RR^2$, whose vector field is shown in \figref{oscillator:a}.
    The continuous-time dynamics are given by the following system of differential equations:
    $\dot{s}(t) = A s(t)$, where $A = \begin{pmatrix} 0 & 1 \\ -1 & 0 \end{pmatrix}$.
    The solution is~$s(t) = e^{At} s_0$ for some initial state~$s_0$.
    For this system, we only have a single (dummy) control action, $\act = \{a\}$.
    Fixing the control period~$t = 1.2$ yields the discrete-time system~$(S, \act, \suc)$ where~$\suc(s, a) \approx \begin{pmatrix} 0.36 & 0.93 \\ -0.93 & 0.36 \end{pmatrix} s$.
\end{example}

\paragraph{Partitioning.}

A \emph{partition}~$\grid \subseteq \powerset{S}$ of~$S$ is a set of pairwise disjoint sets of states (i.e., $\forall \cell_1 \ne \cell_2 \in \grid.\ \cell_1 \cap \cell_2 = \emptyset$) whose union is~$S$ (i.e., $S = \bigcup_{\cell \in \grid} \cell$).
We call the elements~$\cell$ of~$\grid$ \emph{cells}.
For a state~$s \in S$, $[s]_\grid$ is the unique cell~$\cell$ such that~$s \in \cell$.
We furthermore define two helper functions~$\getcellsinner{\cdot} \colon \powerset{S} \to \powerset{\grid}$ and~$\getcellsouter{\cdot} \colon \powerset{S} \to \powerset{\grid}$ to under- and overapproximate a set of states~$X \subseteq S$ with cells: $\getcellsinner{X} = \{\cell \in \grid \mid \cell \subseteq X\}$ maps~$X$ to all its cells that are contained in~$X$, and~$\getcellsouter{X} = \{\cell \in \grid \mid \cell \cap X \ne \emptyset\}$ maps~$X$ to all cells that intersect with~$X$.

A cell~$\cell$ is \emph{axis-aligned} if there exist intervals~$\I_1, \dots, \I_d$ such that~$\cell = \I_1 \times \dots \times \I_d$.
A partition~$\grid$ is axis-aligned if all cells are axis-aligned.
Moreover, $\grid$ is a \emph{regular grid} if for any two cells~$\cell_1 \ne \cell_2$ in~$\grid$ and any dimension~$i = 1, \dots, d$, the diameters in dimension~$i$ are identical.
%
In what follows, we consider axis-aligned regular grid partitions, or \emph{grids} for short.
Grids enjoy properties such as easy representation by just storing the bounds and diameters.

\paragraph{Strategies and safety.}

Given a control system~$(S, \act, \suc)$, a \emph{strategy}~$\strategy \colon S \to \powerset{\act}$ maps a state~$s$ to a set of (allowed) actions~$a$.
(In the special case where~$\strategy \colon S \to \act$ uniquely determines the next action~$a$, we call~$\strategy$ a \emph{controller}.)
A sequence~$\traj = s_0 a_0 s_1 a_1 \dots$ is a \emph{trajectory} of~$\strategy$ if~$a_i \in \strategy(s_i)$ and~$s_{i+1} \in \suc(s_i, a_i)$ for all~$i$.
A safety property~$\safe \subseteq S$ is characterized by the set of safe states.
We call~$\strategy^X$ a \emph{safety strategy}, or \emph{shield},  with respect to a set~$X$
if all trajectories starting from any initial state~$s_0 \in X$ are safe, i.e., only visit safe states.
We often omit the set~$X$.

\medskip

In general, a safety strategy for infinite state spaces~$S$ cannot be effectively computed.
The typical mitigation is to instead compute a safety strategy for a finite-state abstraction.
One common such abstraction is a grid of finitely many cells.
The grid induces a two-player game.
Given a cell~$\cell$, Player~1 challenges with an action~$a \in \act$.
Player~2 responds with a cell~$\cell'$ such that~$\cell \xrightarrow{a} \cell'$.
Player~1 wins if the game continues indefinitely, and Player~2 wins if $\cell' \not \subseteq \safe$.
Solving this game yields a safety strategy over cells, which then induces a safety strategy over the (concrete) states in~$S$ that uses the same behavior for all states in the same cell.
We formalize this idea next.

\paragraph{Labeled transition system.}

Given a control system~$(S, \act, \suc)$, a grid~$\grid \subseteq \powerset{S}$ induces a finite labeled transition system~$(\grid, \act, \to)$ that connects cells via control actions if they can be reached in one step:
\begin{equation}\label{eq:transition}
    \cell \xrightarrow{a} \cell' \iff \exists s \in \cell.\ \suc(s, a) \cap \cell' \ne \emptyset.
\end{equation}


\paragraph{Grid extension.}

Given a grid~$\grid$ and a safety property~$\safe$, the set of \emph{controllable cells}, or simply \emph{safe cells}, is the maximal set of cells~$\controllablecells$ such that
\begin{equation}\label{eq:controllable_cells}
    \controllablecells = \getcellsinner{\safe} \cap \{ \cell \in \grid \mid \exists a \in \act.\ \forall \cell'.\ \cell \xrightarrow{a} \cell' \implies \cell' \in \controllablecells \}.
\end{equation}

It is straightforward to compute~$\controllablecells$ with a finite fixpoint iteration.
If~$\controllablecells$ is nonempty, there exists a safety strategy~$\strategy^{\grid} \colon \grid \to \powerset{\act}$ at the level of cells (instead of concrete states) with respect to the set~$\controllablecells$, where the most permissive such strategy~\cite{BernetJW02} is
\[
    \strategy^{\grid}(\cell) = \{ a \in \act \mid \forall \cell'.\ \cell \xrightarrow{a} \cell' \implies \cell' \in \controllablecells \}.
\]

A safety strategy~$\strategy^{\grid}$ over the grid~$\grid$ induces the safety strategy~$\strategy^X(s) = \strategy^{\grid}([s]_\grid)$ over the original state space~$S$, with respect to the set~$X = \bigcup_{\cell \in \controllablecells} \cell$.
The converse does not hold, i.e., a safety strategy may exist over~$S$ but not over~$\grid$, because the grid introduces an abstraction, as demonstrated next.

\begin{lemma}[\cite{BrorholtJLLS23}]\label{lemma:soundness}
	Let~$(S, \act, \suc)$ be a control system, $\safe \subseteq S$ be a safety property, and $\grid \subseteq \powerset{S}$ be a partition.
	If~$\strategy^{\grid}$ is a safety strategy over cells, then $\strategy^X(s) = \strategy^{\grid}([s]_\grid)$ is a safety strategy over states~$s \in X \subseteq S$, where~$X = \bigcup_{\cell \in \controllablecells} \cell$.
\end{lemma}

Recall that a grid is an abstraction.
The precision of this abstraction is controlled by the size of the grid cells, which we also refer to as the granularity.

\begin{figure}[t]
    \centering
    \begin{subfigure}[t]{0.326\linewidth}
        \begin{tikzpicture}[scale=0.74]
	\fill[fill=wetasphalt] (1,1) rectangle (3,3);
	\draw[draw=none,fill=concrete] (2,2) circle (0.4);
	\draw[step=1] (0,0) grid (4,4);
	%
	\begin{axis}[%
			width    = 55.8mm,
			height   = 55.6mm,
			view     = {0}{90},  
			xlabel   = $x$,
			ylabel   = $y$,
			domain   = -2.2:2.2,
			y domain = -2.2:2.2,
			xtick    = {-2,...,2},
			ytick    = {-2,...,2},
			xmin=-2,xmax=2,ymin=-2,ymax=2
		]
		\addplot3[orange,quiver={u=y/sqrt(x^2+y^2),v=-x/sqrt(x^2+y^2),scale arrows=0.2},samples=8,-stealth] (x,y,0);
	\end{axis}
	\node[circle,fill=nephritis,draw,inner sep=0.5mm] at (3.12,2.5) (start) {};
	\node[circle,fill=peterriver,draw,inner sep=0.5mm] at (2.9,1.15) (end) {};
	\draw[successor,thick,->,out=290,in=55] (start) to (end);
\end{tikzpicture}
        \caption{Original state space.}
        \label{fig:oscillator:a}
    \end{subfigure}
    \hfill
    \begin{subfigure}[t]{0.326\linewidth}
        \begin{tikzpicture}[scale=0.74]
	\fill[fill=wetasphalt] (0,0) rectangle (4,4);
	\draw[draw=none,fill=concrete] (2,2) circle (0.4);
	\draw[step=1] (0,0) grid (4,4);
	%
	\begin{axis}[%
			width    = 55.8mm,
			height   = 55.6mm,
			view     = {0}{90},  
			xlabel   = $x$,
			ylabel   = $y$,
			domain   = -2.2:2.2,
			y domain = -2.2:2.2,
			xtick    = {-2,...,2},
			ytick    = {-2,...,2},
			xmin=-2,xmax=2,ymin=-2,ymax=2
		]
		\addplot3[orange,quiver={u=y/sqrt(x^2+y^2),v=-x/sqrt(x^2+y^2),scale arrows=0.2},samples=8,-stealth] (x,y,0);
	\end{axis}
\end{tikzpicture}
        \caption{After fixpoint iteration.}
        \label{fig:oscillator:b}
    \end{subfigure}
    \hfill
    \begin{subfigure}[t]{0.326\linewidth}
        \begin{tikzpicture}[scale=0.74]
	\fill[fill=wetasphalt] (0,0) rectangle (4,1);
	\draw[draw=none,fill=concrete] (0,0) rectangle (4,0.8);
	\draw[step=1] (0,0) grid (4,4);
	\begin{axis}[%
			width    = 55.8mm,
			height   = 55.6mm,
			view     = {0}{90},  
			xlabel   = $\theta$,
			ylabel   = $r$,
			domain   = -pi:pi,
			y domain = 0:2,
			xtick    = {-pi,0,pi},
			xticklabels = {$-\pi$,0,$\pi$},
			ytick    = {0,...,2},
			xmin=-pi, xmax=pi, ymin=0, ymax=2
		]
		\addplot3[orange,quiver={u=-4,v=0,scale arrows=0.1},samples=8,-stealth] (x,y,0);
	\end{axis}
	\node[circle,fill=emerald,draw,inner sep=0.5mm] at (2.22,2.4) (start) {};
	\node[circle,fill=peterriver,draw,inner sep=0.5mm] at (1.5,2.4) (end) {};
	\draw[successor,thick,->] (start) -- (end);
\end{tikzpicture}
        \caption{Transformed state space.}
        \label{fig:oscillator:c}
    \end{subfigure}
    \caption{\subref{fig:oscillator:a}~Harmonic oscillator in $x/y$ state space with an obstacle~$O$ (gray). Initially, the black cells~$\getcellsouter{O}$ are unsafe. The example state (green) leads to an unsafe cell, rendering its cell unsafe too.
    \subref{fig:oscillator:b}~In the fixpoint, all cells are unsafe.
    \subref{fig:oscillator:c}~Transformation to polar coordinates. The initial marking is also the fixpoint.}
    \label{fig:oscillator}
\end{figure}
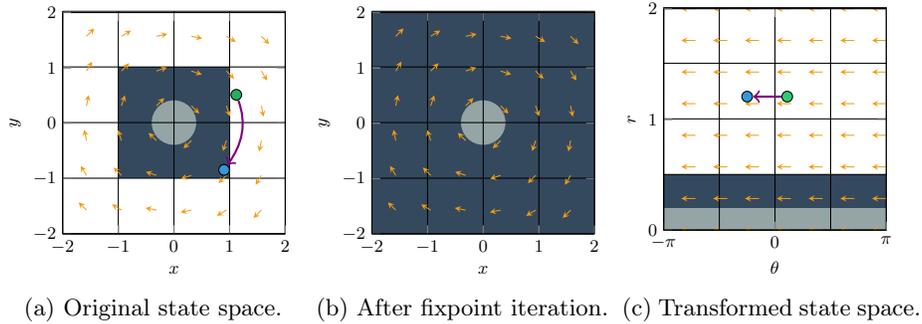

\begin{example}\label{ex:oscillator2}
    Consider again the harmonic oscillator from Example~\ref{ex:oscillator1}.
    To add a safety constraint, we place a disc-shaped obstacle, i.e., a set of states~$O \subseteq S$, in the center.
    Since this system only has a dummy action~$a$, there is a unique strategy~$\strategy$ that always selects this action.
    Clearly, $\strategy$ is safe for all states that do not intersect with the obstacle~$O$ because all trajectories circle the origin.

    With a rectangular grid in the $x/y$ state space (with cell diameter~$1$ in each dimension), we face two fundamental problems.
    The first problem is that, in order to obtain a tight approximation of a disc with a rectangular grid, one requires a fine-grained partition.
    Thus, precision comes with a significant computational overhead.
    Recall that the cells that are initially marked unsafe are given by~$\getcellsouter{O}$, which are drawn black in \figref{oscillator:a}.

    The second problem is similar in nature, but refers to the system dynamics instead.
    Since the trajectories of most systems do not travel parallel to the coordinate axes, a rectangular grid cannot capture the successor relation (and, hence, the required decision boundaries for the strategy) well.
    Consider the state highlighted in green in \figref{oscillator:a}.
    Its trajectory leads to an unsafe cell, witnessing that its own cell is also unsafe.
    By iteratively applying this argument, the fixpoint is $\controllablecells = \emptyset$, i.e., no cell is considered safe (\figref{oscillator:b}).
    This means that, for the chosen grid granularity, no safety strategy~$\strategy^{\grid}$ at the level of cells exists.

    We remark that, since we are only interested in safety at discrete points in time, finer partitions could still yield safety strategies for this example.
\end{example}


%

\section{Shielding in Transformed State Spaces}\label{sec:shielding}


In this section, we show how a transformation of the state space can be used for grid-based shield synthesis, and demonstrate that it can be instrumental.

\subsection{State-Space Transformations}

We recall the principle of state-space transformations.
Consider a state space~$S \subseteq \RR^d$.
A transformation to another state space~$T \subseteq \RR^{d'}$ is any function~$\f \colon S \to T$.

For our application, some transformations are better than others.
We call these transformations \emph{grid-friendly}, where, intuitively, cells in the transformed state space~$T$ are better able to separate the controllable from the uncontrollable states, i.e., capture the decision boundaries well.
This is for instance the case if there is an invariant property and~$\f$ maps this property to a single dimension.

\begin{example}\label{ex:oscillator3}
    Consider again the harmonic oscillator from Example~\ref{ex:oscillator2}.
    We transform the system to a new state space, with the goal of circumventing the two problems identified above.
    Recall that we want to be able to represent a disc shape as well as circular trajectories in a grid-friendly way.
    Observe that the radius of the circle described by a trajectory is an invariant of the trajectory.
    This motivates to choose a transformation from Cartesian coordinates to polar coordinates.
    In polar coordinates, instead of~$x$ and~$y$, we have the dimensions~$\theta$ (angle) and~$r$ (radius).
    The transformation is $\f(x, y) = (\theta, r)^\top = (\atan(y, x), \sqrt{x^2 + y^2})^\top$, and the transformed state space is~$T = \rinterval{-\pi}{\pi} \times \rinterval{0}{\sqrt{8}}$.
    The result after transforming the system, including the obstacle and the two example states, is shown in \figref{oscillator:c}.
    As can be seen, the grid boundaries are parallel to both the obstacle boundaries as well as the dynamics, which is the best-case scenario.
    Observe that the radius dimension ($r$) is invariant.
    Hence, no white cell reaches a black cell and no further cells need to be marked unsafe.
\end{example}

\subsection{Shield Synthesis in a Transformed State Space}

In the following, we assume to be given a control system~$(S, \act, \suc_S)$, a safety property~$\safe$, another state space~$T$, a transformation~$\f \colon S \to T$, and a grid~$\grid \subseteq \powerset{T}$.
Our goal is to compute the controllable cells similar to Eq.~\eqref{eq:controllable_cells}.
However, since the grid is defined over~$T$, we need to adapt the definition.
The set of controllable cells is the maximal set of cells~$\controllablecellstransform$ such that
\begin{equation}\label{eq:controllable_cells_transform}
    \controllablecellstransform = \getcellsinner{f(\safe)} \cap \{ \cell \in \grid \mid \exists a \in \act.\ \forall \cell'.\ \cell \xrightarrow{a} \cell' \implies \cell' \in \controllablecellstransform \}.
\end{equation}

The first change is to map~$\safe$ to cells over~$T$.
Next, it is convenient to define a new control system~$(T, \act, \suc_T)$ that imitates the original system in the new state space.
The new successor function~$\suc_T \colon T \times \act \to \powerset{T}$ is given indirectly as
\begin{equation}\label{eq:successor_transformed_implicit}
	\suc_T(f(s), a) = \f(\suc_S(s, a)).
\end{equation}

The second change in Eq.~\eqref{eq:controllable_cells_transform} is implicit in the transition relation~$\cell \xrightarrow{a} \cell'$ of the labeled transition system~$(\grid, \act, \to)$.
Recall from Eq.~\eqref{eq:transition} that the transitions are defined in terms of the successor function~$\suc_T$:
\begin{equation*}
    \cell \xrightarrow{a} \cell' \iff \exists t \in \cell.\ \suc_T(t, a) \cap \cell' \ne \emptyset.
\end{equation*}

\paragraph{State-based successor computation.}

To simplify the presentation, for the moment, we only consider a single state~$t \in T$.
To effectively compute its successors, we cannot directly use Eq.~\eqref{eq:successor_transformed_implicit} because it starts from a state~$s \in S$ instead.
Hence, we first need to map~$t$ back to~$S$ using the \emph{inverse transformation}~$\finv \colon T \to \powerset{S}$, defined as~$\finv(t) = \{ s \in S \mid \f(s) = t \}$.
The resulting set is called the \emph{preimage}.

Now we are ready to compute~$\suc_T(t, a)$ for any state~$t \in T$ and action~$a \in \act$.
First, we map~$t$ back to its preimage~$X = \finv(t)$.
Second, we apply the original successor function~$\suc_S$ to obtain~$X' = \suc_S(X, a)$.
Finally, we obtain the corresponding transformed states~$Y = \f(X')$.
In summary, we have
\begin{equation}\label{eq:suc_transformed}
    \suc_T(t, a) = \f(\suc_S(\finv(t), a)).
\end{equation}

Note that, if the transformation~$\f$ is bijective, its inverse~$\finv$ is deterministic and we have~$\finv(\f(s)) = s$ and~$\f(\finv(t)) = t$ for all~$s \in S$ and~$t \in T$.

\begin{example}
    Consider again the harmonic oscillator from Example~\ref{ex:oscillator2}.
    The inverse transformation is~$\finv(\theta, r) = \twovec{r \cos(\theta)}{r \sin(\theta)}$.
    The blue successor state of the green state in \figref{oscillator:c} is computed by mapping to the green state in \figref{oscillator:a} via~$\finv$, computing the blue successor state via~$\suc$, and mapping back via~$\f$.
\end{example}

\paragraph{Grid-based successor computation.}

\begin{figure}[t]
    \centering
    \begin{subfigure}[b]{0.31\linewidth}
        \centering
        \begin{tikzpicture}[t/.style={->,thick,>=stealth}]
	\node (C) {$\cell$};
	\node[right=of C] (Y) {$Y$};
	\node[below=of C] (X) {$X$};
	\node[at=(Y|-X)] (X') {$X'$};
	%
	\draw[t] (C) to node[above] {$\suc_T$} (Y);
	\draw[t] (C) to node[left] {$\finv$} (X);
	\draw[t] (X) to node[above] {$\suc_S$} (X');
	\draw[t] (X') to node[right] {$\f$} (Y);
\end{tikzpicture}
        \caption{Commutative diagram.}
        \label{fig:commutative_diagram}
    \end{subfigure}
    \hfill
    \begin{subfigure}[b]{0.68\linewidth}
        \centering
        \hspace*{3mm} $S$ \hspace*{37mm} $T$ \\
        \includesvg[width=\linewidth]{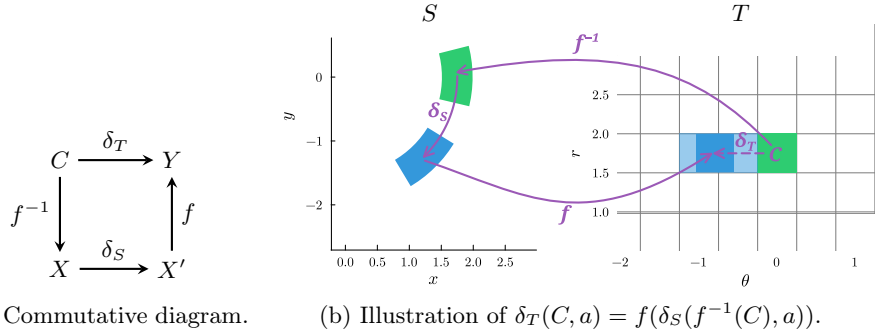}
        \caption{Illustration of $\suc_T(\cell, a) = \f(\suc_S(\finv(\cell), a))$.}
        \label{fig:transformation_successor_illustration}
    \end{subfigure}
    \caption{The successor function~$\suc_T$ for the cell~$\cell$ (green) in the transformed state space~$T$ is computed in three steps.
    First, we map to the original state space~$S$ via~$\finv$.
    Second, we compute the successors via~$\suc_S$.
    Third, we map back to the transformed state space~$T$ via~$\f$ (dark blue).
    Finally, we can identify all cells intersecting with this set via~$\getcellsouter{\cdot}$ (light blue).}
    \label{fig:transformation_successor}
\end{figure}

Eq.~\eqref{eq:suc_transformed} is directly applicable to cells via set extension, and no further modification is required.
%
%
We provide illustrations of the construction in \figref{transformation_successor}.

The construction allows us to compute sound shields, both in the transformed and in the original state space.

\begin{theorem}\label{thm:soundness}
	Let~$(S, \act, \suc)$ be a control system, $\safe \subseteq S$ be a safety property, $\f \colon S \to T$ be a transformation with inverse~$\finv \colon T \to \powerset{S}$, and~$\grid \subseteq \powerset{T}$ be a partition of~$T$.
	Define the control system~$(T, \act, \suc_T)$ with $\suc_T$ according to $f$.
	Let~$\controllablecellstransform$ be the set of controllable cells of $\grid$ and~$\strategy^{\grid}$ be a corresponding  safety strategy over cells.
	Then the following are safety strategies over states:
	\begin{itemize}
		\item[$\bullet$] $\strategy^{Y}(t) = \strategy^{\grid}([t]_\grid)$ over states in~$t \in Y \subseteq T$, where~$Y = \bigcup_{\cell \in \controllablecellstransform} \cell$.
		
		\item[$\bullet$] $\strategy^X(s) = \strategy^{\grid}([f(s)]_\grid)$ over states in~$s \in X \subseteq S$, where~$X = \finv(\bigcup_{\cell \in \controllablecellstransform} \cell)$.
	\end{itemize}
\end{theorem}

\begin{proof}
    We first need to argue that~$\suc_T$ is well-defined.
    If~$\f$ is not injective, then~$\finv$ is nondeterministic, i.e., generally yields a set of states, but the set extension of~$\suc_S$ treats this case.
    If~$\f$ is not surjective, its inverse is undefined for some states~$t \in T$.
    Note that the set extension ignores these states: for any set~$Y \subseteq T$ we have~$\finv(Y) = \{s \in S \mid \f(s) \in Y\}$.
    In particular, if no state in~$\cell$ has a preimage, $\suc_T(\cell, a) = \emptyset$.
    Thus, $\suc_T$ is well-defined.

    The first claim follows directly from Lemma~\ref{lemma:soundness}.
    For the second claim, fix any state~$s \in X$ and~$a \in \strategy^X(s)$.
    We need to show that all states in~$\suc_S(s, a)$ are safe, i.e., in~$X$ as well.
    By construction, $f(s) \in \bigcup_{\cell \in \controllablecellstransform} \cell$ and $a \in \strategy^{\grid}([f(s)]_\grid)$.
    Hence,
    \begin{equation}\label{eq:suc_contained}
        \suc_T(f(s), a) \subseteq \bigcup_{\cell \in \controllablecellstransform} \cell.
    \end{equation}

    We also use the following simple lemma:
    \begin{equation}\label{eq:monotonicity}
        \forall s' \in S.\ s \in \finv(f(s')).
    \end{equation}

    Finally, we get (applying monotonicity in the last inclusion):
    \begin{align*}
        \suc_S(s, a)
        &\stackrel{\eqref{eq:monotonicity}}{\subseteq} \suc_S(\finv(f(s)), a)
        \stackrel{\eqref{eq:monotonicity}}{\subseteq} \finv(\f(\suc_S(\finv(f(s)), a))) \\
        &\stackrel{\eqref{eq:suc_transformed}}{\subseteq} \finv(\suc_T(f(s), a))
        \stackrel{\eqref{eq:suc_contained}}{\subseteq} \finv\left(\bigcup\nolimits_{\cell \in \controllablecellstransform} \cell\right)
        = X. \tag*{\qed}
    \end{align*}
\end{proof}

Note that we obtain Lemma~\ref{lemma:soundness} as the special case where~$\f$ is the identity.

\subsection{Shielding and Learning}

We assume the reader is familiar with the principles of reinforcement learning.
Here we shortly recall from~\cite{BrorholtJLLS23} how to employ~$\strategy^{\grid}$ for safe reinforcement learning.
The input is a Markov decision process (MDP) and a reward function, and the output is a controller maximizing the expected cumulative return.
The MDP is a model with probabilistic successor function~$\suc_P \colon S \times \act \times S \to [0, 1]$.
An MDP induces a control system~$(S, \act, \suc_S)$ with nondeterministic successor function~$\suc_S(s, a) = \{s' \in S \mid \suc_P(s, a, s') > 0\}$ as an abstraction where the distribution has been replaced by its support.

Now consider \figref{shielding}, which integrates a transformed shield into the learning process.
In each iteration, the shield removes all unsafe actions (according to~$\strategy^{\grid}$) from the agent's choice.
By construction, when starting in a controllable state, at least one action is available, and all available actions are guaranteed to lead to a controllable state again.
Thus, by induction, all possible trajectories are infinite and never visit an unsafe state.
Furthermore, filtering unsafe actions typically improves learning convergence because fewer options need to be explored.

\paragraph{Learning in~$S$ and~$T$.}

Recall from Theorem~\ref{thm:soundness} that we can apply the shield both in the transformed state space and in the original state space by using the transformation function~$\f$.
This allows us to also perform the learning in either state space.
We consider the following setup the default: learning in the original state space~$S$ under a shield computed in the transformed state space~$T$.

An alternative is to directly learn in~$T$.
A potential motivation could be that learning, in particular agent representation, may also be easier in~$T$.
For instance, the learning method implemented in \uppaalstratego represents an agent by axis-aligned hyperrectangles~\cite{JaegerJLLST19}.
Thus, a grid-friendly transformation may also be beneficial for learning, independent of the shield synthesis.
We will investigate the effect in our experiments.

\section{Experiments}\label{sec:experiments}

In this section, we demonstrate the benefits of state-space transformations for three models.\footnote{A repeatability package is available here: \\ \url{https://github.com/AsgerHB/state-space-transformation-shielding}.}
For the first two models, we use domain knowledge to select a suitable transformation.
For the third model, we instead derive a transformation experimentally.
The implementation builds on our synthesis method~\cite{BrorholtJLLS23}.

\subsection{Satellite Model}

\begin{figure}[t]
    \begin{minipage}{0.40\linewidth}
        \centering
        \includesvg[width=0.75\textwidth]{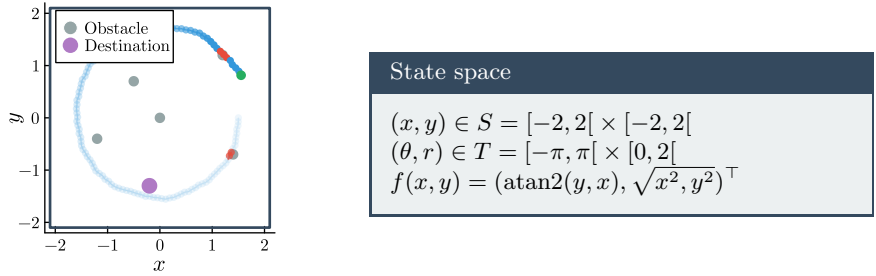}
    \end{minipage}
    \hfill
    \begin{minipage}{0.55\linewidth}
        \begin{statespacebox}
            $(x, y) \in S = \rinterval{-2}{2}\times\rinterval{-2}{2}$

            $(\theta, r) \in T = \rinterval{-\pi}{\pi}\times\rinterval{0}{2}$

            $\f(x, y) = (\atan(y, x), \sqrt{x^2, y^2})^\top$
        \end{statespacebox}
    \end{minipage}
    \caption{Satellite model.}
    \label{fig:satellite_unsafe_trace}
\end{figure}

For the first case study, we extend the harmonic oscillator with two more control actions to also move inward and outward: $\act = \{\textit{ahead}, \textit{out}, \textit{in}\}$.
The box to the side shows the relevant information about the transformation.
Compared to Example~\ref{ex:oscillator3}, beside the actions, we modify two parts.
First, the transformed state space~$T$ is reduced in the radius dimension to~$r \in \rinterval{0}{2}$ because values outside the disc with radius~$2$ are not considered safe (see below).
Second, the successor function still uses matrix~$A$ from Example~\ref{ex:oscillator1} but with a control period of~$t = 0.05$.
The successor function thus becomes $\suc(s, a) = e^{At}\twovec{rc \cos(\theta)}{rc \sin(\theta)}$, where for~$s = (x, y)^\top$ and~$\f$ as in Example~\ref{ex:oscillator3} we have
\begin{equation*}
    \left ( \begin{matrix}
        \theta \\ r 
    \end{matrix} \right )
    = f(s),
    \quad
    c = \begin{cases}
        0.99 & \text{if } a = \textit{in} \\
        1.01 & \text{if } a = \textit{out} \\
        1 & \text{otherwise.}
    \end{cases},
    \quad
    e^{At} \approx \left ( \begin{matrix}
        1.00 & 0.05\\
        -0.05 & 1.00 \\
    \end{matrix} \right )
\end{equation*}

Instead of one large obstacle, we add several smaller stationary (disc-shaped) obstacles.
The shield has two goals: first, the agent must avoid a collision with the obstacles; second, the agent's distance to the center must not exceed~$2$.
\figref{satellite_unsafe_trace} shows the size and position of the obstacles (gray).
Overlaid is a trajectory (blue) produced by a random agent that selects actions uniformly.
Some states of the trajectory collide with obstacles (red).

Additionally, we add an optimization component to the system.
A disc-shaped \emph{destination} area (purple) spawns at a random initial position (inside the 2-unit circle).
Colliding with this area grants a reward and causes it to reappear at a new position.
The optimization criterion for the agent is thus to visit as many destinations as possible during an episode.

\medskip

\begin{figure}[t]
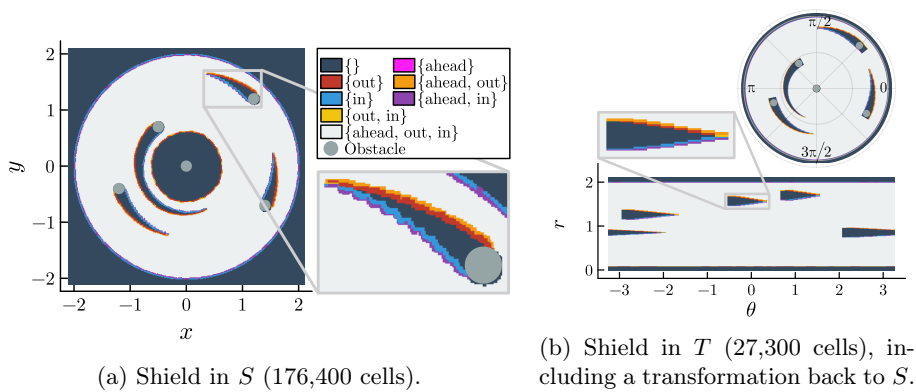

    \centering
    \begin{subfigure}[b]{0.57\linewidth}
        \includesvg[width=\linewidth]{Graphics/Spiral/Spiral_Standard_State_Space.svg}
        \caption{Shield in $S$ (176,400 cells).}
        \label{fig:satellite_shield_original}
    \end{subfigure}
    \hfill
    \begin{subfigure}[b]{0.41\linewidth}
        \includesvg[width=\linewidth]{Graphics/Spiral/Spiral_Altered_State_Space.svg}
        \caption{Shield in $T$ (27,300 cells), including a transformation back to $S$.}
        \label{fig:satellite_shield_transformed}
    \end{subfigure}
    \caption{Shields for the satellite model. The legend applies to both figures.}
    \label{fig:satellite_shields}
\end{figure}

\figref{satellite_shield_original} shows a shield obtained in the original state space.
First, we note that a fine grid granularity is required to accurately capture the decision boundaries.
In particular, the ``tails'' behind the obstacles split into regions where moving \textit{ahead} is no longer possible.
There is a small region at the tip of the tail (yellow) where the agent may either move \textit{in} or \textit{out}, but not \textit{ahead} anymore.

Moreover, despite this high precision, the obstacle in the center causes a large set of cells around it to be marked unsafe, although we know that the \textit{ahead} (and also \textit{out}) action is safe.
This is a consequence of the abstraction in the grid.

\medskip

Now we transform the system, for which we choose polar coordinates again.
\figref{satellite_shield_transformed} shows a shield obtained in this transformed state space.
As we saw for the harmonic oscillator, the boundary condition is well captured by a grid.
The obstacles also produce ``tails'' in this transformation, which require relatively high precision in the grid to be accurately captured.
Still, since the shapes are axis-aligned, and the size of the transformed state space is different, the number of cells can be reduced by one order of magnitude. 
The grid over the original state space had $176{,}400$~cells, compared to $27{,}300$~cells in the transformed state space.
Computing the original shield took $2$~minutes and $41$~seconds, while computing the transformed shield only took $10$~seconds.
Finally, the region marked unsafe at the bottom of \figref{satellite_shield_transformed}, which corresponds to the central obstacle in the original state space, is tight, unlike in \figref{satellite_shield_original}.
In summary, the transformed shield is both easier to compute and more precise.

\subsection{Bouncing-Ball Model}

\begin{figure}[t]
    \begin{minipage}{0.40\linewidth}
        \centering
        \includesvg[width=0.75\textwidth]{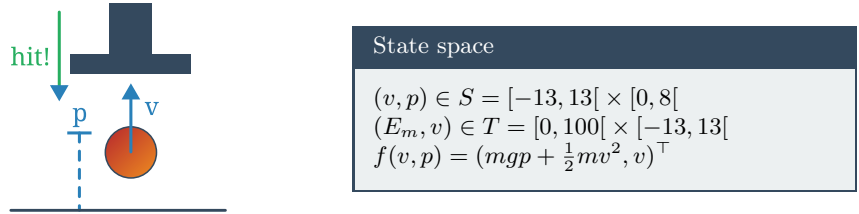}
    \end{minipage}
    \hfill
    \begin{minipage}{0.55\linewidth}
        \begin{statespacebox}
            $(v, p) \in S = \rinterval{-13}{13}\times\rinterval{0}{8}$

            $(E_m, v) \in T = \rinterval{0}{100}\times\rinterval{-13}{13}$

            $\f(v, p) = (m g p + \frac{1}{2} m v^2, v)^\top$
        \end{statespacebox}
    \end{minipage}
    \caption{Bouncing-ball model.}
    \label{fig:bouncing_ball}
\end{figure}

\begin{figure}[t]
    \centering
    \begin{subfigure}[b]{0.49\linewidth}
        \includegraphics[width=\linewidth]{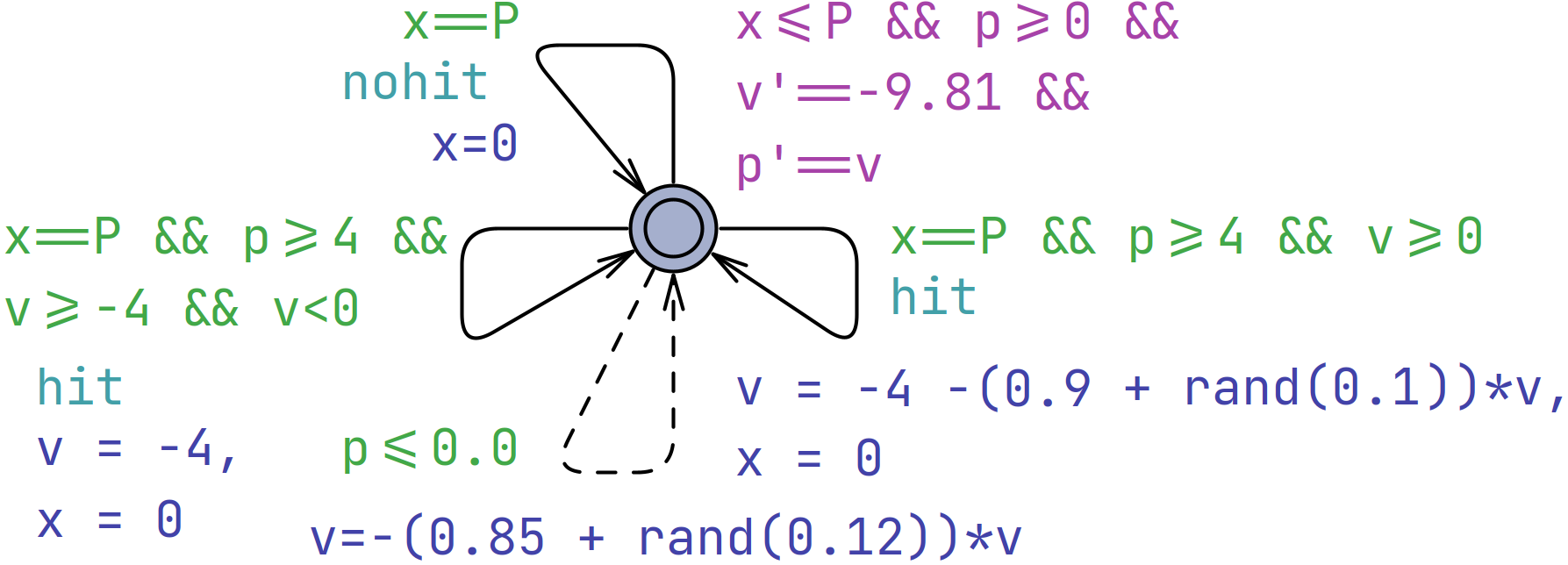}
        \vspace*{2mm}
        \caption{Hybrid automaton.}
        \label{fig:bb_automaton}
    \end{subfigure}
    \hfill
    \begin{subfigure}[b]{0.49\linewidth}
        \includesvg[width=\linewidth]{Graphics/BouncingBall/BB_Standard_State_Space.svg}
        \caption{Shield in~$S$ (520,000 cells).}
        \label{fig:bb_shield_original}
    \end{subfigure}
    \\
    \begin{subfigure}[b]{0.49\linewidth}
        \includesvg[width=\linewidth]{Graphics/BouncingBall/BB_Altered_State_Space.svg}
        \caption{Shield in~$T$ (650 cells).}
        \label{fig:bb_shield_transformed}
    \end{subfigure}
    \hfill
    \begin{subfigure}[b]{0.49\linewidth}
        \includesvg[width=\linewidth]{Graphics/BouncingBall/BB_Altered_State_Space_in_Standard_State_Space.svg}
        \caption{Shield in~$T$ transformed back to~$S$.}
        \label{fig:bb_shield_transformed_projection}
    \end{subfigure}
    \caption{Hybrid automaton and shields for the bouncing-ball model.}
    \label{fig:bb_shield}
\end{figure}

For the second case study, we consider the model of a bouncing ball from~\cite{BrorholtJLLS23}.
\figref{bouncing_ball} shows an illustration of the system, while \figref{bb_automaton} shows the hybrid-automaton model.
The state space consists of the velocity~$v$ and the position~$p$ of the ball.
When the ball hits the ground, it loses energy subject to a stochastic dampening (dashed transition).
The periodic controller is modeled with a clock~$x$ with implicit dynamics~$\dot{x} = 1$ and control period~$P = 0.1$.
The available actions are~$\act = \{\textit{nohit}, \textit{hit}\}$, where the \textit{nohit} action has no effect and the \textit{hit} action pushes the ball downward subject to its velocity, but only provided it is high enough ($p \ge 4$).

The goal of the shield is to keep the ball bouncing indefinitely, which is modeled as nonreachability of the set of states~$p \le 0.01 \land |v| \le 1$.

The optimization task is to use the \textit{hit} action as rarely as possible, which is modeled by assigning it with a cost and minimizing the total cost.

\medskip

Despite its simple nature, this model has quite intricate dynamics, including stochastic and hybrid events that require zero-crossing detection, which makes determining reachability challenging.
It was shown in~\cite{BrorholtJLLS23} that a sampling-based shield synthesis is much more scalable than an approach based on guaranteed reachability analysis ($19$~minutes compared to $41$~hours).
The grid needs to be quite fine-grained to obtain a fixpoint where not every cell is marked unsafe.
This corresponds to~$520{,}000$ cells, and the corresponding shield is shown in \figref{bb_shield_original}.

\medskip

Now we use a transformation to make the shield synthesis more efficient.
The mechanical energy~$E_m$ stored in a moving object is the sum of its potential energy and its kinetic energy, respectively.
Formally, $E_m(p, v) = m g p + \frac{1}{2} m v^2$, where~$m = 1$ is the mass and~$g = 9.81$ is gravity.
Thus, the mechanical energy of a ball in free fall (both with positive or negative velocity) remains invariant.
Hence, $E_m$ is a good candidate for a transformation.

However, only knowing~$E_m$ is not sufficient to obtain a permissive shield because states with the same value of~$E_m$ may be below or above~$p = 4$ and hence may or may not be hit.
The equation for~$E_m$ depends on both~$p$ and~$v$.
In this case, it is sufficient to know only one of them.
Here, we choose the transformed state space~$T$ with just~$E_m$ and~$v$.
The transformation function is~$\f(v, p) = (m g p + \frac{1}{2} m v^2, v)^\top$ and its inverse is~$\finv(E_m, v) = ((E_m - \frac{1}{2} m v^2) / (m g), v)^\top$.
We note that using~$E_m$ and~$p$ instead yields a shield that marks all cells unsafe.
This is because~$v$ is quadratic in~$E_m$ and, thus, we cannot determine its sign.

\figref{bb_shield_transformed} shows the shield obtained in this transformed state space.
It can be seen that this results in a very low number of just~$650$ cells in total, which is a reduction by three orders of magnitude.
This shield can be computed in just $1.3$~seconds, which compared to $19$~minutes is again a reduction by three orders of magnitude.

To provide more intuition about how precise this shield still is, we project the shield back to the original state space in \figref{bb_shield_transformed_projection}.
While a direct comparison is not fair because the grid granularity differs vastly, overall the shapes are similar.



\subsection{Cart-Pole Model}

\begin{figure}[t]
    \begin{minipage}{0.40\linewidth}
        \includesvg[width=0.75\textwidth]{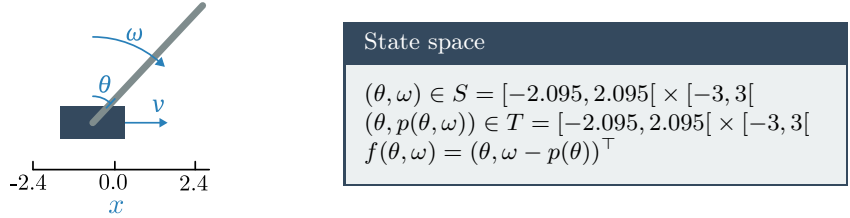}
        \centering
    \end{minipage}
    \hfill
    \begin{minipage}{0.55\linewidth}
        \begin{statespacebox}
            $(\theta, \omega) \in S = \rinterval{-2.095}{2.095}\times\rinterval{-3}{3}$

            $(\theta, p(\theta, \omega)) \in T = \rinterval{-2.095}{2.095}\times\rinterval{-3}{3}$

            $\f(\theta, \omega) = (\theta, \omega - p(\theta))^\top$
        \end{statespacebox}
    \end{minipage}
    \caption{Cart-pole model.}
    \label{fig:cart_pole}
\end{figure}

For the third case study, we consider a model of an inverted pendulum installed on a cart that can move horizontally.
This model is known as the cart-pole model.
An illustration is shown in \figref{cart_pole}.
The dynamics are given by the following differential equations~\cite{Florian05}:
\begin{align*}
    \dot{\theta} &= \omega
    &
    \dot{\omega} &= \frac{\displaystyle g \sin(\theta)+ \cos(\theta) \cdot \left(\frac{-F - m_p \ell \omega^2 \sin(\theta)}{m_c + m_p}\right)}{\displaystyle \ell \left(\frac{4}{3} - \frac{m_p \cos^2(\theta)}{{m_c+m_p}}\right)} \\
    \dot{x} &= v
    &
    \dot{v} &= \frac{F + m_p \ell \left(\omega^2 \sin(\theta) - \dot{\omega} \cos(\theta)\right)}{m_c + m_p}
\end{align*}

The state dimensions are the pole's angle~$\theta$ and angular velocity~$\omega$ as well as the cart's position~$x$ and velocity~$v$.
Moreover, $g = 9.8m/s^2$ is gravity, $\ell = 0.5$\,m is the pole's length, $m_p = 0.1$\,kg is the pole's mass, and~$m_c = 1$\,kg is the cart's mass.
Finally, $F = \pm 10$ is the force that is applied, corresponding to the action from~$\act = \{\textit{left}, \textit{right}\}$, which can be changed at a rate of~$0.02$ (control period).

The goal of the shield is to balance the pole upright, which translates to the condition that the angle stays in a small cone~$|\theta| \le 0.2095$.

The optimization goal is to keep the cart near its initial position~$x(0)$.
Moving more than $2.4$\,m away yields a penalty of~$1$ and resets the cart.



\medskip

Observe that the property for the shield only depends on the pole and not on the cart.
Hence, it is sufficient to focus on the pole dimensions~$\theta$ and~$\omega$ for shield synthesis, and leave the cart dimensions for the optimization.
A shield in the original state space is shown in \figref{cart_pole_shield_original}.

\medskip

In the following, we describe a state-space transformation for shield synthesis.
Unlike for the other models, we are not aware of an invariant property that is useful for our purposes.
Instead, we will derive a transformation in two steps.

Recall that a transformation is useful if a grid in the new state space captures the decision boundaries well, i.e., the new decision boundaries are roughly axis-aligned.
Thus, our plan is to approximate the shape of the decision boundaries in the first step and then craft a suitable transformation in the second step.

\subsubsection{Approximating the Decision Boundaries.}

\begin{figure}[t]
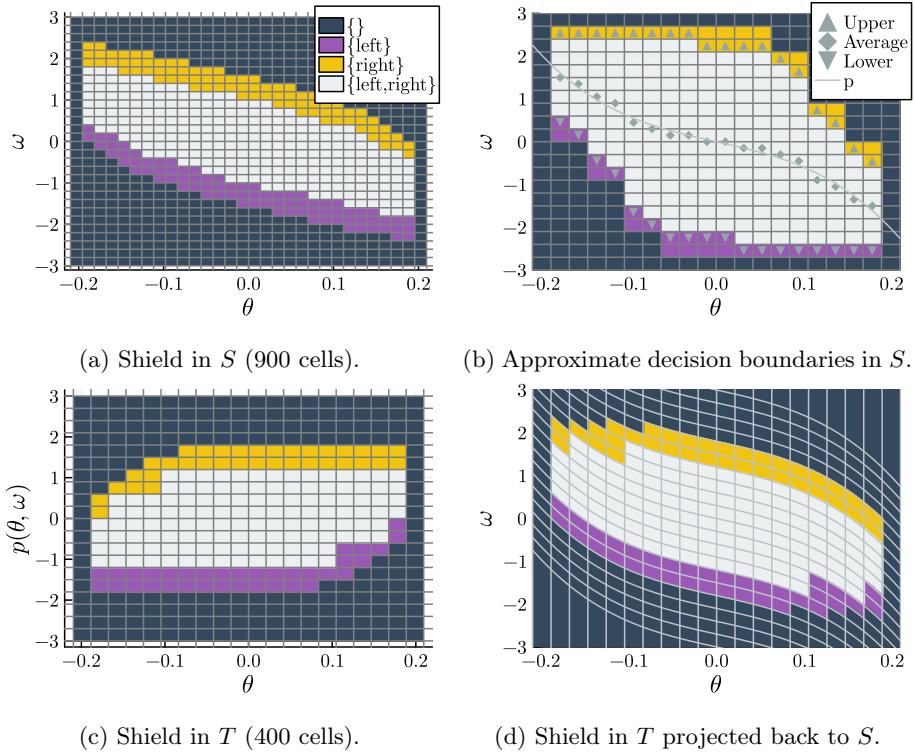

    \centering
    \begin{subfigure}[b]{0.49\linewidth}
        \includesvg[width=\linewidth]{Graphics/CartPole/CartPole_Standard_State_Space.svg}
        \caption{Shield in~$S$ ($900$~cells).}
        \label{fig:cart_pole_shield_original}
    \end{subfigure}
    \hfill
    \begin{subfigure}[b]{0.49\linewidth}
        \includesvg[width=\linewidth]{Graphics/CartPole/Standard_State_Space_Fitting_Polynomial.svg}
        \caption{Approximate decision boundaries in~$S$.}
        \label{fig:cart_pole_shield_unfinished}
    \end{subfigure}
    \\
    \begin{subfigure}[b]{0.49\linewidth}
        \includesvg[width=\linewidth]{Graphics/CartPole/CartPole_Altered_State_Space.svg}
        \caption{Shield in~$T$ ($400$~cells).}
        \label{fig:cart_pole_shield_transformed}
    \end{subfigure}
    \hfill
    \begin{subfigure}[b]{0.49\linewidth}
        \includesvg[width=\linewidth]{Graphics/CartPole/CartPole_Altered_State_Space_in_Standard_State_Space.svg}
        \caption{Shield in~$T$ projected back to~$S$.}
        \label{fig:cart_pole_shield_transformed_projected}
    \end{subfigure}
    \caption{Shield computation for the cart-pole model. The legend in \figref{cart_pole_shield_original} applies to all subfigures.}
    \label{fig:cart_pole_shield}
\end{figure}

\figref{cart_pole_shield_original} shows the decision boundaries of a fixpoint computed using $30 \times 30$ cells.
However, our work of state-space transformations was motivated because computing the shield is generally not feasible in the first place.

Therefore, here we take a different approach, which uses a grid of just~$20 \times 20$ cells.
Computing a shield for such a coarse grid in the original state space yields~$\controllablecells = \emptyset$, i.e., all cells become unsafe.
%
%
This is a consequence of the abstraction, i.e., a trajectory at the grid level may be spurious at the state level.
This abstraction grows with the number of steps of the trajectory.
Our idea is thus to only perform the fixpoint iteration at the grid level for a low number (here: three) of steps.
(Technically, this means that the strategy is only guaranteed to be safe for three steps.)
The result is the marking of cells in \figref{cart_pole_shield_unfinished}.
Indeed, the decision boundaries roughly approximate those in \figref{cart_pole_shield_original}.


\subsubsection{Crafting a Transformation.}

We want to find a grid-friendly transformation that ``flattens out'' the decision boundaries.
Our idea is to keep the dimension~$\theta$ and replace~$\omega$ by a transformation that is ``flatter.''
We observe that the upper (yellow) and lower (purple) decision boundaries are symmetric.
Hence, the distance to the average of the upper and lower boundaries is a good approximation.

This idea is visualized in \figref{cart_pole_shield_unfinished}.
Here we compute the average (diamonds) of the upper and lower boundaries (triangles).
Then we fit a polynomial to approximate this shape.
In our implementation, we used the Julia \href{https://github.com/JuliaMath/Polynomials.jl}{\texttt{Polynomials}} library, which implements a standard linear least squares method~\cite{DraperS98}.
Here, a third-degree polynomial~$p(\theta) = - 141.6953 \cdot \theta^3 - 4.5508 \cdot \theta$ is sufficient.

To obtain the full transformation, we need to express the offset from~$p(\theta)$.
Thus, we choose~$\f(\theta, \omega) = (\theta, \omega - p(\theta))^\top$.
The inverse function is~$\finv(\theta, z) = (\theta, z + p(\theta))^\top$, where~$z$ is the new dimension in the transformed state space ($T$).

\medskip

The resulting shield is shown in \figref{cart_pole_shield_transformed}.
The grid size is $400$~cells, as compared to $900$~cells in the original state space.
Both took less than a second to synthesize, at $244$\,ms and $512$\,ms, respectively.

\subsection{Strategy Reduction}

\begin{table}[t]
    \caption{Representation sizes of the computed shields.}
    \label{tab:shield_reduction}
    \centering
    \begin{tabular}{l @{~~} c @{~~} r @{~~} r}
    \toprule
    \textbf{Model}                 & \textbf{State space} & \textbf{Number of cells} & \textbf{Number of nodes} \\
    \midrule
    \multirow{2}{*}{Satellite}     & $S$ & 176,400   & 4,913  \\
                                   & $T$ & 27,300    & 544  \\
    \midrule
    \multirow{2}{*}{Bouncing ball} & $S$ & 520,000   & 940  \\
                                   & $T$ & 650       & 49  \\
    \midrule
    \multirow{2}{*}{Cart-pole}     & $S$ & 900       & 99  \\
                                   & $T$ & 400       & 32  \\
    \bottomrule
    \end{tabular}%
\end{table}

We provide an overview of the savings due to computing the shield in the transformed state space in Table~\ref{tab:shield_reduction}.
The column labeled \emph{Number of cells} clearly shows a significant reduction in all cases.
We remark that, in order to have a fair comparison, we have selected the grid sizes from visual inspection to ensure that the plots look sufficiently close.
However, it is not the case that one of the shields is more permissive than the other.

The strategies above can be represented with a $d$-dimensional matrix.
Matrices are inherently limiting representations of shields, especially when the shield should be stored on an embedded device.
Empirically, a decision tree with axis-aligned predicates is a much better representation.
To demonstrate the further saving potential, we converted the shields to decision trees and additionally applied the reduction technique from~\cite{HoegPetersenLWJ23}.
The last column in Table~\ref{tab:shield_reduction} shows the number of nodes in the decision trees.
As can be seen, we always achieve another significant reduction by one to two orders of magnitude.

    \subsection{Shielded Reinforcement Learning}

\begin{table}[t]
    \caption{Cumulative return over $1000$ episodes with both shielding and learning in either of the state spaces.
    Higher return is better for the satellite model, and vice versa for the other models.
    Each row's best result is marked in bold face.}
    \label{tab:shielded_learning}
    \centering
    \begin{tabular}{c @{~~~} r @{~~} r @{~~} r c @{~~~} r @{~~} r @{~~} r c @{~~~} r @{~~} r @{~~} r}
        \toprule
        \textbf{Learning} & \multicolumn{3}{c}{\textbf{Satellite} ($\nearrow$)} && \multicolumn{3}{c}{\textbf{Bouncing ball} ($\searrow$)} && \multicolumn{3}{c}{\textbf{Cart-pole} ($\searrow$)} \\
        \midrule
        & \multicolumn{1}{c}{None} & \multicolumn{1}{c}{$S$} & \multicolumn{1}{c}{$T$} && \multicolumn{1}{c}{None} & \multicolumn{1}{c}{$S$} & \multicolumn{1}{c}{$T$} && \multicolumn{1}{c}{None} & \multicolumn{1}{c}{$S$} & \multicolumn{1}{c}{$T$} \\
        \cline{2-4} \cline{6-8} \cline{10-12}
        $S$ & 1.123 & 0.786 & \textbf{1.499} && 39.897 & 37.607 & \textbf{36.593} && 0.007 & 0.019 & \textbf{0.001} \\
        $T$ & 0.917 & 0.889 & \textbf{1.176} && 39.128 & 40.024 & \textbf{39.099} && \textbf{0.000} & \textbf{0.000} & \textbf{0.000} \\
        \bottomrule
    \end{tabular}%
\end{table}

The only motivation for applying a state-space transformation was to be able to compute a cheaper shield.
From the theory, we cannot draw any conclusions about the impact on the controller performance.
We investigate this impact in the following experiments, with the main result that the transformed shield actually increases the performance consistently.

\medskip

We conduct six experiments for each of the three models.
For shielding, we consider three variants (no shield, shielding in the original state space~$S$, and shielding in the transformed state space~$T$).
For each variant, we reinforcement-learn two controllers.
One controller is trained in the original state space~$S$, while the other controller is trained in the transformed state space~$T$.

In Table~\ref{tab:shielded_learning}, we provide the learning results for all combinations of shielding and learning.
The data is given as the cumulative return obtained over $1000$ executions of the environment and the respective learned agent using \uppaal.

The results show that, for all models, the highest reward is achieved by the  controller operating under the shield in the transformed state space.
This holds regardless of which state space the controller was trained in.
Additionally, the controller that was trained in the original state space achieves higher performance.
Thus, the transformation was not helpful for the learning process itself.


%
%

\section{Conclusion}\label{sec:conclusion}

We have demonstrated that state-space transformations hold great potential for shield synthesis.
We believe that they are strictly necessary when applying shield synthesis to many practical systems due to state-space explosion.

In the first two case studies, we used domain knowledge to select a suitable transformation.
In the third case study, we instead engineered a transformation in two steps.
We plan to generalize these steps to a principled method and investigate how well it applies in other cases.

State-space transformations can be integrated with many orthogonal prior extensions of grid-based synthesis.
One successful extension is, instead of precomputing the full labeled transition system, to compute its transitions on the fly~\cite{HsuMMS18b}.
Another extension is the multilayered abstraction~\cite{GirardGM16,HsuMMS18a}.
Going one step further, in cases where a single perfect transformation does not exist, we may still be able to find a family of transformations of different strengths.

\subsection*{Acknowledgments}

We thank Tom Henzinger for the suggestion to study level sets.
This research was partly supported by the Independent Research Fund Denmark under reference number 10.46540/3120-00041B, DIREC - Digital Research Centre Denmark under reference number 9142-0001B, and the Villum Investigator Grant S4OS under reference number 37819.

\bibliographystyle{splncs04}
\bibliography{bibliography}

\end{document}